\documentclass[a4paper,11pt]{article}
\usepackage[utf8]{inputenc}
\usepackage{a4wide}
\usepackage{latexsym,amsfonts,amsmath,amssymb,mathrsfs,url,amsthm}
\usepackage{mathtools}
\usepackage{color,graphicx}
\usepackage{lipsum}
\usepackage{xcolor}
\usepackage{hyperref}
\usepackage[normalem]{ulem}
\providecommand{\keywords}[1]
{
  \noindent \small	
  \textbf{Keywords:} #1
}

\providecommand{\amscode}[1]
{
  \noindent \small	
  \textbf{AMS subject classifications:} #1
}
\newtheorem{theorem}{Theorem}[section]
\newtheorem{definition}[theorem]{Definition}
\newtheorem{lemma}[theorem]{Lemma}
\newtheorem{corollary}[theorem]{Corollary}
\newtheorem{remark}[theorem]{Remark}

\newcommand{\re}{\mathrm{e}}
\newcommand{\ri}{\mathrm{i}}
\newcommand{\rd}{\, \mathrm{d}}

\newcommand{\bsk}{\boldsymbol{k}}
\newcommand{\bsell}{\boldsymbol{\ell}}
\newcommand{\bszero}{\boldsymbol{0}}

\newcommand{\bsx}{\boldsymbol{x}}
\newcommand{\bsy}{\boldsymbol{y}}
\newcommand{\bsz}{\boldsymbol{z}}
\newcommand{\ran}{\mathrm{ran}}
\newcommand{\wor}{\mathrm{wor}}

\newcommand{\EE}{\mathbb{E}}
\newcommand{\II}{\mathbb{I}}
\newcommand{\NN}{\mathbb{N}}
\newcommand{\PP}{\mathbb{P}}
\newcommand{\RR}{\mathbb{R}}
\newcommand{\ZZ}{\mathbb{Z}}
\newcommand{\Acal}{\mathcal{A}}

\DeclareMathOperator{\supp}{supp}
\DeclareMathOperator{\width}{width}

\allowdisplaybreaks

\title{Tractability results for integration in subspaces of the Wiener algebra
\thanks{The work of J.~D.\ is supported by ARC grant DP220101811. 
The work of T.~G.\ is supported by JSPS KAKENHI Grant Number 23K03210.
The work of K.~S.\ is supported by JSPS KAKENHI Grant Number 24K06857.}
}
\author{
Josef Dick\thanks{School of Mathematics and Statistics, The University of New South Wales, Kensington, NSW 2052, Australia ({\tt josef.dick@unsw.edu.au})}, 
Takashi Goda\thanks{Graduate School of Engineering, The University of Tokyo, 7-3-1 Hongo, Bunkyo-ku, Tokyo 113-8656, Japan ({\tt goda@frcer.t.u-tokyo.ac.jp})}, 
Kosuke Suzuki\thanks{Faculty of Science, Yamagata University, 1-4-12 Kojirakawa-machi, Yamagata, 990-8560, Japan ({\tt kosuke-suzuki@sci.kj.yamagata-u.ac.jp})}
}
\date{\today}

\begin{document}

\maketitle

\sloppy

\begin{abstract}
In this paper, we present some new (in-)tractability results related to the integration problem in subspaces of the Wiener algebra over the $d$-dimensional unit cube. We show that intractability holds for multivariate integration in the standard Wiener algebra in the deterministic setting, in contrast to polynomial tractability in an unweighted subspace of the Wiener algebra recently shown by Goda (2023). Moreover, we prove that multivariate integration in the subspace of the Wiener algebra introduced by Goda is strongly polynomially tractable if we switch to the randomized setting, where we obtain a better $\varepsilon$-exponent than the one implied by the standard Monte Carlo method. We also identify subspaces in which multivariate integration in the deterministic setting are (strongly) polynomially tractable and we compare these results with the bound which can be obtained via Hoeffding's inequality.
\end{abstract}

\amscode{41A55, 42B05, 65D30, 65D32}

\section{Introduction}
\subsection{Integration and tractability}
This paper is concerned with tractability of multivariate  numerical integration problems of functions defined over the $d$-dimensional unit cube. 

For an integrable function $f: [0,1)^d\to \RR$, we denote its integral by
\[ I_d(f)=\int_{[0,1)^d}f(\bsx)\rd \bsx. \]
In the deterministic setting, we approximate $I_d(f)$ using a general algorithm of the form
\begin{align}\label{eq:general_quadrature}
    Q_{d,n}(f) = \varphi_n(f(\bsx_1),\ldots,f(\bsx_n)),
\end{align}
with the nodes $\bsx_1,\ldots,\bsx_n\in [0,1)^d$, where $\varphi_n: \RR^n\to \RR$ is a (linear or nonlinear) mapping.
Here, for $h\geq 2$, each node $\bsx_h$ can be chosen sequentially based on the previously evaluated function values, i.e., $\bsx_h = \psi_h(f(\bsx_1),\ldots,f(\bsx_{h-1}))$ for some mapping $\psi_h$.
If $\varphi_n$ is a linear mapping and the nodes are chosen non-adaptively, the algorithm \eqref{eq:general_quadrature} reduces to a non-adaptive linear algorithm:
\begin{align}\label{eq:quadrature}
    Q_{d,n}(f)=\sum_{h=0}^{n-1}c_h f(\bsx_h)
\end{align}
with a set of fixed $n$ points $P_{d,n}=\{\bsx_0,\ldots,\bsx_{n-1}\}\subset [0,1)^d$ and associated real coefficients $\{c_0,\ldots,c_{n-1}\}.$ A quasi-Monte Carlo (QMC) rule is a special case of \eqref{eq:quadrature}, in which all the coefficients $c_h$ are equal to $1/n$.

The worst-case error of an algorithm $Q_{d,n}$ in a Banach space $F$ of integrable functions $f: [0,1)^d\to \RR$ with norm $\|\cdot\|$ is defined as
\[ e^{\wor}(F,Q_{d,n}):=\sup_{f\in F, \|f\|\leq 1}\left| I_d(f)-Q_{d,n}(f)\right|. \] 
Information-based complexity \cite{NW08, NW10, TWW88} is concerned with how many function evaluations are necessary to make the worst-case error less than or equal to a tolerance $\varepsilon\in (0,1)$. 
This quantity is called the \emph{information complexity} and is defined by
\[ n^{\wor}(\varepsilon, d, F):= \min\{ n\in \NN\, \mid \, \exists Q_{d,n},\, e^{\wor}(F,Q_{d,n})\leq \varepsilon \}.\]
Here we note that, since the unit ball of $F$ is always symmetric and convex, it follows from, for instance, \cite[Theorem~4.7]{NW08}, that a non-adaptive linear algorithm \eqref{eq:quadrature} minimizes the worst-case error among general algorithms of the form \eqref{eq:general_quadrature}.

In the randomized setting, we approximate $I_d(f)$ by a randomized algorithm $A_{d}$ that is defined as a pair of a probability space $(\Omega,\Sigma,\mu)$ and a collection of deterministic algorithms $(Q_{d,\omega})_{\omega\in \Omega}$. That is, for each fixed $\omega\in \Omega$, the corresponding algorithm $Q_{d,\omega}$ is a deterministic algorithm of the form \eqref{eq:general_quadrature}. We assume that $(Q_{d,\omega}(f))_{\omega\in \Omega}$ is measurable for each $f\in F$. In this setting, one of the quality measures for $A_d$ is the worst-case randomized  error:
\[ e^{\ran}(F,A_d):=\sup_{f\in F, \|f\|\leq 1}\EE_{\omega}\left[ \left| I_d(f)-Q_{d,\omega}(f)\right| \right].\]
Here, the number of points used by $Q_{d, \omega}(f)$ can change depending on $\omega\in \Omega$ and $f\in F$, and we denote it by $n(f,\omega)$. The cardinality of a randomized algorithm $A_d$ is defined by
\[ \#(A_d) := \sup_{f\in F}\EE_{\omega}\left[ n(f,\omega) \right] ,\]
see \cite[Section~4.3.3]{NW08}.
Similarly to the deterministic setting as above, the information complexity in the randomized setting is defined by
\[ n^{\ran}(\varepsilon, d, F):= \min\{ n\in \NN\, \mid \, \exists A_d,\,  e^{\ran}(F,A_d)\leq \varepsilon\: \: \text{and}\: \: \#(A_d)\leq n\}.\]

Tractability is a key terminology in the field of information-based complexity. Although there are many other or refined notations \cite{NW08,NW10}, we introduce the major three notions of tractability that are used in this paper. 

\begin{definition}
    Multivariate integration in the deterministic setting for a Banach space $F$ is 
    \begin{itemize}
        \item \emph{intractable} if $n(\varepsilon,d,F)$ grows at least exponentially fast in $\varepsilon^{-1}$ or $d$ for $\varepsilon\to 0$ and $d\to \infty$, i.e. $\lim_{\varepsilon^{-1} + d \to \infty} \frac{\ln n(\varepsilon,d,F)}{\varepsilon^{-1}+d} > 0$;
        \item \emph{polynomially tractable} if there exist $c,\alpha,\beta>0$ such that 
        \[ n(\varepsilon,d,F)\leq c \varepsilon^{-\alpha}d^\beta \]
        holds for any $\varepsilon\in (0,1)$ and $d\in \NN$, where $\alpha$ and $\beta$ are referred to as the $\varepsilon$-exponent and $d$-exponent, respectively;
        \item and \emph{strongly polynomially tractable} if there exist $c,\alpha>0$ such that 
        \[ n(\varepsilon,d,F)\leq c \varepsilon^{-\alpha} \]
        holds for any $\varepsilon\in (0,1)$ and $d\in \NN$, where $\alpha$ is referred to as the $\varepsilon$-exponent.
    \end{itemize}
    These notions are also applied to the randomized setting by replacing $n(\varepsilon,d,F)$ with $n^{\ran}(\varepsilon, d, F)$.
\end{definition}

\subsection{Subspaces of Wiener algebra and known results}

In this paper, we focus on the information complexity in (subspaces of) the Wiener algebra:
\begin{align}\label{eq:subspace_wiener}
    F_{d,r}:=\left\{ f\in C([0,1)^d)\: \middle| \: \|f\|:=\sum_{\bsk\in \ZZ^{d}}|\hat{f}(\bsk)|\, r(\bsk)<\infty \right\},
\end{align}
where $\hat{f}(\bsk)$ denotes the $\bsk$-th Fourier coefficient of $f$, i.e.,
\[ \hat{f}(\bsk) = \int_{[0,1)^d}f(\bsx)\, \re^{-2\pi \ri \bsk\cdot \bsx}\rd \bsx,\]
and $r(\bsk)\geq 1$ represents the ``Fourier weight'' for the corresponding frequency. Notice that $F_{d,r}$ coincides with the standard Wiener algebra if $r(\bsk)=1$ for all $\bsk\in \ZZ^d$. 

In \cite{D14}, one of the current authors proved polynomial tractability (in the deterministic setting) for the intersection of the standard Wiener algebra and a class of H\"{o}lder continuous functions. 
Later, in \cite{G23}, another of the current authors proved a similar result on polynomial tractability for the subspace of the Wiener algebra given by
\[ r(\bsk) \geq r_1(\bsk) := \max\left(1,\log \min_{j\in \supp(\bsk)}|k_j|\right),\]
where $\supp(\bsk):=\{j\in \{1,\ldots,d\}\mid k_j\neq 0\}$. More precisely, it has been shown that an upper bound $n^{\wor}(\varepsilon, d, F_{d,r_1})\leq c \varepsilon^{-3}d^3$ holds for a constant $c>0$. 
It is important to emphasize that these function spaces $F$ are \emph{unweighted} in the sense that all input variables contribute equally to the norm. Therefore, for any permutation matrix $\pi$ and $f\in F$, it holds that $f\circ \pi \in F$ and $\|f\circ \pi\|=\|f\|$. 

We refer to \cite{CH24, K23, KV23} for more recent progress on this line of research. 
For instance, in contrast to the above positive results, intractability in the deterministic setting has been shown in \cite{KV23} for the function space
\[ \left\{ f\in C([0,1)^d)\: \middle| \: \|f\|^2:=\sum_{\bsk\in \ZZ^{d}}|\hat{f}(\bsk)|^2\, \max\left(1,\log \max_{j\in \supp(\bsk)}|k_j|\right)<\infty \right\}.\]
We also note that the space \eqref{eq:subspace_wiener}, with special forms of $r$, has been studied recently, for instance, in \cite{JUV23,KPUU23} in the context of $L_p$-approximation.

In passing, both results in \cite{D14} and \cite{G23} are proven using constructive arguments, that is, explicit QMC rules using Korobov's $p$-set or its multiset union over different cardinalities are constructed to attain the desired worst-case error bounds.
The weighted version of the star discrepancy for Korobov's $p$-set was studied in \cite{DP15}, where the weighted star discrepancy is shown to be bounded polynomially in $d$ or even independent of $d$ under certain summability conditions on the weight parameters. 

\subsection{Summary of our results}

Regarding the multivariate integration problem, the exact tractability characterization for the space introduced in \cite{G23} remains open. Although it was proven to be at least polynomially tractable, it could be strongly polynomially tractable. 
While we do not answer this problem in this paper, we present the following related results.

\begin{enumerate}
\item (Theorem~\ref{thm:main1}) If $r(\bsk)= 1$ for all $\bsk\in \ZZ^d$,
the information complexity in the deterministic setting is infinite for $\varepsilon\in (0,1/2)$,
concluding that the multivariate integration problem for the standard Wiener algebra is intractable. 

\item (Theorem~\ref{thm:main2}) If $r(\bsk) \geq r_2(\bsk):=\max\left(\width(\supp(\bsk)),\log \min_{j\in \supp(\bsk)}|k_j|\right)$, where we define $\width(\supp(\bsk)):=\max_{j\in \supp(\bsk)} j-\min_{j\in \supp(\bsk)}j+1$, we show by the same QMC rule as in \cite{G23} that the multivariate integration problem in the deterministic setting for $F_{d, r}$ is strongly polynomially tractable. We point out the slight difference between $r_1(\bsk)$ and $r_2(\bsk)$: in $r_1(\bsk)$, the first argument of the maximum is always $1$, while in $r_2(\bsk)$, the first argument is $\width(\supp(\bsk))$.

\item (Theorem~\ref{thm:main4}) Many classical results, see, e.g., \cite{DKP22}, use Fourier weights of product form. Here we show, again by the same QMC rule as in \cite{G23}, that for the product Fourier weights $r(\bsk) \geq r_3(\bsk) = \prod_{j=1}^d \max(1, \log|k_j|)$, the multivariate integration problem in the deterministic setting for $F_{d, r}$ is at least polynomially tractable.

\item (Theorem~\ref{thm:main5}) For comparison, we use a proof technique from \cite{HNWW01} to show that the multivariate integration problem in the deterministic setting for $F_{d,r}$ with $r(\bsk) \geq r_4(\bsk) = \prod_{j=1}^d \max(1, |k_j|)$, is polynomially tractable. 

\item (Theorem~\ref{thm:main3}) For the case $r(\bsk) \geq r_1(\bsk)$, we show by an explicit randomized algorithm that the multivariate integration problem in the randomized setting for $F_{d,r}$ is strongly polynomially tractable. Here we obtain a better $\varepsilon$-exponent than what is expected from the standard Monte Carlo method.
\end{enumerate}

Regarding the second result, although the space $F_{d,r_2}$ is weighted in the sense that all input variable do not contribute equally to the norm, it remains invariant under the reversion of the variables, i.e., if $f\in F_{d,r_2}$, then we have $g\in F_{d,r_2}$ and $\|f\|=\|g\|$ where $g(x_1,\ldots,x_d)=f(x_d,\ldots,x_1)$. This contrasts many existing results on strong polynomial tractability for multivariate integration in the deterministic setting, where ``coordinate weights'' are introduced to model the relative importance of each group of variables, and variables are typically assumed ordered in decreasing order of importance, see \cite{DGPW17,DP15,NW08,NW10,SW98} among many others. In fact, it seems not possible to characterize the space $F_{d,r_2}$ in such a way. 

Regarding the third and fourth results, the space $F_{d, r_4}$ is embedded in the space $F_{d, r_3}$, and hence Theorem~\ref{thm:main4} (proven by an explicit algorithm) provides a de-randomized version of Theorem~\ref{thm:main5} (proven as an existence result). However, the bound in Theorem~\ref{thm:main5} is stronger and so both results have some advantages.

\section{Lower bound}\label{sec:intractable}
As announced, we prove intractability for the standard Wiener algebra.
\begin{theorem}\label{thm:main1}
Consider the standard Wiener algebra, i.e., the space \eqref{eq:subspace_wiener} with $r(\bsk)= 1$ for all $\bsk\in \ZZ^d$. The minimal worst-case error for general algorithms of the form \eqref{eq:general_quadrature} equals $1/2$ for any $n,d\geq 1$, which implies that the information complexity in the deterministic setting, $n^{\wor}(\varepsilon, d, F_{d,r})$, is infinite for any $\varepsilon\in (0,1/2)$.
\end{theorem}

\begin{proof} 
Recall that, for any deterministic algorithm of the form \eqref{eq:general_quadrature}, there exists a non-adaptive linear algorithm of the form \eqref{eq:quadrature} with no larger worst-case error; see \cite[Theorem 4.7]{NW08}. Thus it suffices to
prove a lower bound on the worst-case error that holds for any linear quadrature rule. Since our argument below is based on construction of a univariate fooling function, let us focus on the case with $d=1$.

Recalling that our quadrature rule is now restricted to the form \eqref{eq:quadrature}, if $\sum_{h=0}^{n-1}|c_h|\leq 1/2$ holds, for a constant function $g\equiv 1$, we have $I_1(g)=1, \|g\|=1,$ and also
\[ Q_{1,n}(g)=\sum_{h=0}^{n-1}c_h\leq \sum_{h=0}^{n-1}|c_h|\leq \frac{1}{2}.\]
This means that the worst-case error of any linear quadrature rule satisfying $\sum_{h=0}^{n-1}|c_h|\leq 1/2$ is bounded below by
\begin{align*}
    e^{\wor}(F_{1,r},Q_{1,n})\geq \left| I_1(g)-Q_{1,n}(g)\right|\geq \frac{1}{2},
\end{align*}
independently of $n$. Otherwise, if $\sum_{h=0}^{n-1}|c_h|> 1/2$ holds, it follows from the simultaneous variant of Dirichlet's approximation theorem that, for any real numbers $x_0,x_1,\ldots,x_{n-1}\in [0,1)$, $\rho\geq 2$, and
\[ M=\left\lceil \frac{2\pi}{\arccos\left(1-(\rho\sum_{h=0}^{n-1}|c_h|)^{-1}\right)}\right\rceil^n,\] 
there exist integers $p_0,p_1,\ldots,p_{n-1}$ and $q\in \{1,\ldots,M\}$ such that
\[ \left| q x_h - p_h \right| <\frac{1}{M^{1/n}}=\frac{\arccos(1-(\rho\sum_{h=0}^{n-1}|c_h|)^{-1})}{2\pi}\quad \text{for all $h = 0, 1, \ldots, n-1$}.\]

Let $\epsilon_h := q x_h - p_h$, for $h = 0, 1, \ldots, n-1$. With this bound, for the function $g^*(x)=1-\cos(2\pi qx)$, we have $I_1(g^*)=1, \|g^*\|=2$ and
\begin{align*}
    |Q_{1,n}(g^*)| & = \left|\sum_{h=0}^{n-1}c_h (1-\cos(2\pi qx_h))\right| = \left|\sum_{h=0}^{n-1}c_h (1-\cos(2\pi (p_h+\epsilon_h)))\right| \\
    & = \left|\sum_{h=0}^{n-1}c_h (1-\cos(2\pi \epsilon_h))\right| \leq \sum_{h=0}^{n-1}|c_h|\, |1-\cos(2\pi \epsilon_h)|\\
    & \leq \sum_{h=0}^{n-1}|c_h|\, \left|1-\cos\left(\arccos\left(1-\left(\rho\sum_{h'=0}^{n-1}|c_{h'}|\right)^{-1}\right)\right)\right|\\
    & = \sum_{h=0}^{n-1}|c_h|\left(\rho\sum_{h'=0}^{n-1}|c_{h'}|\right)^{-1}=\frac{1}{\rho}.
\end{align*}
This means that the worst-case error of any linear quadrature algorithm satisfying $\sum_{h=0}^{n-1}|c_h|> 1/2$ is bounded below by
\begin{align*}
    e^{\wor}(F_{1,r},Q_{1,n})\geq \frac{\left| I_1(g^*)-Q_{1,n}(g^*)\right|}{\|g^*\|}\geq \frac{1-1/\rho}{2}.
\end{align*}
It is important to note that this lower bound holds for arbitrarily large $\rho$ and also applies to the case with $d\geq 2$ since our univariate function $g^*$ belongs to $F_{d,r}$ for any $d\geq 2$.

We now show a matching upper bound. For any $n,d\geq 1$, let us consider an algorithm
\[ Q_{d,n}(f)=Q_{d,n}^*(f):=\frac{f(\bszero)}{2},\]
which only uses a single function evaluation at the origin $\bsx=\bszero$. For any $f \in F_{d,r}$ we have
\begin{align*}
\left|I(f) - Q_{d,n}^*(f)\right|
= \left|\hat{f}(\bszero) - \frac{1}{2}\sum_{\bsk \in \ZZ^d}\hat{f}(\bsk) \right|
&= \left|\frac{1}{2}\hat{f}(\bszero) - \frac{1}{2}\sum_{\bsk \in \ZZ^d \setminus \{\bszero\}}\hat{f}(\bsk) \right|\\
&\le \frac{1}{2}|\hat{f}(\bszero)| + \frac{1}{2}\sum_{\bsk \in \ZZ^d \setminus \{\bszero\}}|\hat{f}(\bsk)|
= \frac{1}{2}\|f\|.
\end{align*}
This proves $e^{\wor}(F_{d,r},Q_{d,n}^*) \leq 1/2$.
Recalling the lower bound shown above, regarding the minimal worst-case error for general algorithms, we have obtained 
\[ \inf_{Q_{d,n}}e^{\wor}(F_{d,r},Q_{d,n})=\frac{1}{2}.\]
for any $n$ and $d$,
which implies that $n^{\wor}(\varepsilon, d, F_{d,r})$ is infinite for any $\varepsilon\in (0,1/2).$ This completes the proof.
\end{proof}


\section{(Strong) Tractability in the deterministic setting}\label{sec:str_tractable_det}

\subsection{Numerical integration in $F_{d,r_2}$}

This section is devoted to proving the following upper bound.
\begin{theorem}\label{thm:main2}
Consider the space \eqref{eq:subspace_wiener} with $r(\bsk) \geq r_2(\bsk) = \max(\width (\supp(\bsk)), \log \min_{j \in \supp(\bsk)} |k_j|)$. There exists a constant $C>0$ such that, for any $d\in \NN$ and $\varepsilon \in (0,1)$, we have
    \[ n^{\wor}(\varepsilon, d, F_{d, r})\leq C \varepsilon^{-3}/(\log \varepsilon^{-1}). \]
\end{theorem}

Our proof is constructive in the sense that we provide an explicit QMC rule that attains the desired worst-case error bound. The QMC rule considered here is exactly the same as the one discussed in \cite{G23}. For an integer $m\geq 2$, let
\[ \PP_m := \{\lceil m/2\rceil<p\leq m\, \mid\, \text{$p$ is prime} \}.\]
It is known that there exist constants $c_{\PP}$ and $C_{\PP}$ with $0<c_{\PP}<\min(1,C_{\PP})$ such that
\begin{align}\label{eq:prime_num}
c_{\PP}\frac{m}{\log m}\leq |\PP_m|\leq C_{\PP}\frac{m}{\log m}, \end{align}
for all $m\geq 2$, see \cite[Corollaries~1--3]{RS62}. Now, given an integer $m\geq 2$, we define two different point sets as multiset unions:
\[ P_{d,m}^1=\bigcup_{p\in \PP_m}S_{d,p}\quad \text{and}\quad P_{d,m}^2=\bigcup_{p\in \PP_m}T_{d,p},\]
where $S_{d,p}=\{\bsx_h^{(p)}\mid 0\leq h<p^2\}$ and $T_{d,p}=\{\bsy_{h,\ell}^{(p)}\mid 0\leq h,\ell<p\}$ are sets with $p^2$ points known as \emph{Korobov's $p$-sets} \cite{D14,DP15,HW81,K63}. The points are defined as follows:
\[ \bsx_h^{(p)}=\left( \left\{ \frac{h}{p^2}\right\}, \left\{ \frac{h^2}{p^2}\right\},\ldots, \left\{ \frac{h^d}{p^2}\right\}\right),\]
and
\[  \bsy_{h,\ell}^{(p)}=\left( \left\{ \frac{h\ell}{p}\right\}, \left\{ \frac{h\ell^2}{p}\right\},\ldots, \left\{ \frac{h\ell^d}{p}\right\}\right), \]
respectively, where we write $\{x\}=x-\lfloor x\rfloor$ to denote the fractional part of a non-negative real number $x$. It is important to note that taking a multiset unions of Korobov's $p$-sets with different primes $p$ is crucial in our error analysis. Obviously we have
\[ |P_{d,m}^1|=|P_{d,m}^2|=\sum_{p\in \PP_m}p^2 \ge \left(\frac{m}{2}\right)^2 |\PP_m|.\]

The following result on the exponential sums refines the known results from \cite[Lemmas~4.5 \& 4.6]{HW81} as well as \cite[Lemmas~4.4 \& 4.5]{DP15}.
\begin{lemma}\label{lem:exponential_sum}
    Let $d\in \NN$ and $p$ be a prime.  
    For any $\bsk\in \ZZ^d\setminus \{\bszero\}$ such that there exists at least one index $j^*\in \{1,\ldots,d\}$ where $k_{j^*}$ is not divisible by $p$, i.e., $p\nmid \bsk$, the following bounds hold:
    \[ \left|\frac{1}{p^2}\sum_{h=0}^{p^2-1}\exp\left( 2\pi \ri \bsk\cdot \bsx_h^{(p)}\right)\right| \leq \frac{\width(\supp(\bsk))}{p},\]
    and
    \[ \left|\frac{1}{p^2}\sum_{h,\ell=0}^{p-1}\exp\left( 2\pi \ri \bsk\cdot \bsy_{h,\ell}^{(p)}\right)\right| \leq \frac{\width(\supp(\bsk))}{p}.\]
\end{lemma}

\begin{proof}
    Let us consider the first bound. As we have $\{0,\ldots,p^2-1\}=\{h_0+h_1p \mid 0\leq h_0,h_1<p\}$ and, for each pair of $h_0,h_1\in \{0,\ldots,p-1\}$, it holds that
    \begin{align*}
        \exp\left( 2\pi \ri \bsk\cdot \bsx_{h_0+h_1p}^{(p)}\right) & = \exp\left( \frac{2\pi \ri}{p^2}\sum_{j\in \supp(\bsk)}k_j (h_0+h_1p)^j\right) \\
        & = \exp\left( \frac{2\pi \ri}{p^2}\sum_{j\in \supp(\bsk)}k_j \sum_{a=0}^{j}\binom{j}{a}h_0^a(h_1p)^{j-a}\right) \\
        & = \exp\left( \frac{2\pi \ri}{p^2}\sum_{j\in \supp(\bsk)}k_j (h_0^j+j h_0^{j-1}h_1p)\right),
    \end{align*}
    and hence
    \begin{align*}
        \left|\frac{1}{p^2}\sum_{h=0}^{p^2-1}\exp\left( 2\pi \ri \bsk\cdot \bsx_h^{(p)}\right)\right| & = \left|\frac{1}{p^2}\sum_{h_0,h_1=0}^{p-1}\exp\left( \frac{2\pi \ri}{p^2}\sum_{j\in \supp(\bsk)}k_j (h_0^j+j h_0^{j-1}h_1p)\right)\right| \\
        & = \left|\frac{1}{p}\sum_{h_0=0}^{p-1}\exp\left( \frac{2\pi \ri}{p^2}\sum_{j\in \supp(\bsk)}k_j h_0^j\right)\frac{1}{p}\sum_{h_1=0}^{p-1}\exp\left(\frac{2\pi \ri h_1}{p}\sum_{j\in \supp(\bsk)}k_j j h_0^{j-1}\right)\right| \\
        & \leq \frac{1}{p}\sum_{h_0=0}^{p-1}\left|\frac{1}{p}\sum_{h_1=0}^{p-1}\exp\left(\frac{2\pi \ri h_1}{p}\sum_{j\in \supp(\bsk)}k_j j h_0^{j-1}\right)\right|\\
        & = \frac{1}{p}\sum_{\substack{h_0=0\\ \sum_{j\in \supp(\bsk)}k_j j h_0^{j-1}\equiv 0 \pmod p}}^{p-1}1,
    \end{align*}
    where the last equality follows from the well-known character property for the trigonometric functions, see, e.g., \cite[Lemma~4.3]{DHP15}. Here, by denoting $j_{\min}=\min_{j\in \supp(\bsk)}j$ and $j_{\max}=\max_{j\in \supp(\bsk)}j$, we have
    \[ \sum_{j\in \supp(\bsk)}k_j j h_0^{j-1}=\sum_{\substack{j=j_{\min}\\ j\in \supp(\bsk)}}^{j_{\max}}k_j j h_0^{j-1}= h_0^{j_{\min}-1}\sum_{\substack{j=j_{\min}\\ j\in \supp(\bsk)}}^{j_{\max}}k_j j h_0^{j-j_{\min}}.\]
    As the last sum over $j$ is a polynomial in $h_0$ with degree $j_{\max}-j_{\min}$, the number of solutions of the congruence $\sum_{j\in \supp(\bsk)}k_j j h_0^{j-1}\equiv 0 \pmod p$ is at most $j_{\max}-j_{\min}+1=\width(\supp(\bsk))$. Thus the result follows. Since the second bound can be proven in the same manner, we omit the details.
\end{proof}

Note that, if $k_j$ is divisible by $p$ for all $j$, i.e., $p\mid \bsk$, then we only have a trivial bound on the exponential sum, which is 1. Using this refined result, we obtain the following bounds on the exponential sums for our point sets $P_{d,m}^1$ and $P_{d,m}^2$.
\begin{corollary}\label{cor:exponential_sum}
    Let $d\in \NN$ and $m\geq 2$. 
    For any $\bsk\in \ZZ^d\setminus \{\bszero\}$, it holds that
    \[ \left|\frac{1}{|P_{d,m}^1|}\sum_{p\in \PP_m}\sum_{h=0}^{p^2-1}\exp\left( 2\pi \ri \bsk\cdot \bsx_h^{(p)}\right)\right| \leq \frac{1}{m}\left( 4\width(\supp(\bsk))+\frac{8}{c_{\PP}}\min_{j\in \supp(\bsk)}\log |k_j|\right),\]
    and
    \[ \left|\frac{1}{|P_{d,m}^2|}\sum_{p\in \PP_m}\sum_{h,\ell=0}^{p-1}\exp\left( 2\pi \ri \bsk\cdot \bsy_{h,\ell}^{(p)}\right)\right| \leq \frac{1}{m}\left( 4\width(\supp(\bsk))+\frac{8}{c_{\PP}}\min_{j\in \supp(\bsk)}\log |k_j|\right).\]
\end{corollary}

\begin{proof}
    The following proof for the first bound is similar to that of \cite[Corollary~2.3]{G23}, and the second bound can be proven in a similar way, so we omit the details. Using Lemma~\ref{lem:exponential_sum}, for $\bsk\in \ZZ^d\setminus \{\bszero\}$, we have
    \begin{align*}
        \left|\frac{1}{|P_{d,m}^1|}\sum_{p\in \PP_m}\sum_{h=0}^{p^2-1}\exp\left( 2\pi \ri \bsk\cdot \bsx_h^{(p)}\right)\right| & \leq \frac{1}{|P_{d,m}^1|}\sum_{p\in \PP_m}\left|\sum_{h=0}^{p^2-1}\exp\left( 2\pi \ri \bsk\cdot \bsx_h^{(p)}\right)\right| \\
        & \leq \frac{1}{|P_{d,m}^1|}\sum_{\substack{p\in \PP_m\\ p\nmid \bsk}}p\width(\supp(\bsk))+\frac{1}{|P_{d,m}^1|}\sum_{\substack{p\in \PP_m\\ p\mid \bsk}}p^2\\
        & \leq \frac{m|\PP_m|}{|P_{d,m}^1|}\width(\supp(\bsk))+\frac{m^2}{|P_{d,m}^1|}\sum_{\substack{p\in \PP_m\\ p\mid \bsk}}1\\
        & \leq \frac{m|\PP_m|}{(m/2)^2|\PP_m|}\width(\supp(\bsk))+\frac{m^2}{(m/2)^2|\PP_m|}\sum_{\substack{p\in \PP_m\\ p\mid \bsk}}1\\
        & \leq \frac{4}{m}\width(\supp(\bsk))+\frac{4\log m}{c_{\PP}m}\sum_{\substack{p\in \PP_m\\ p\mid \bsk}}1,
    \end{align*}
    where the last inequality follows from \eqref{eq:prime_num}. To give a bound on the last sum over $p\in \PP_m$ which divide $\bsk$, we use the fact that, for any integers $k,n\in \NN$, $k$ has at most $\log_n k$ prime divisors larger than or equal to $n$. With $\II(\cdot)$ denoting the indicator function, for any index $j^*\in \supp(\bsk)$, we get
    \begin{align*}
        \sum_{\substack{p\in \PP_m\\ p\mid \bsk}}1 & = \sum_{p\in \PP_m}\prod_{j\in \supp(\bsk)}\II(p\mid k_j)\leq \sum_{p\in \PP_m}\II(p\mid k_{j^*}) \leq \log_{\lceil m/2\rceil +1}|k_{j^*}|\leq \frac{2\log |k_{j^*}|}{\log m}.
    \end{align*}
    Since this inequality applies to any index $j^*\in \supp(\bsk)$, it holds that
    \[ \sum_{\substack{p\in \PP_m\\ p\mid \bsk}}1 \leq \frac{2}{\log m}\min_{j\in \supp(\bsk)}\log |k_j|.\]
    This completes the proof.
\end{proof}

Now we are ready to prove Theorem~\ref{thm:main2}.
\begin{proof}[Proof of Theorem~\ref{thm:main2}]
Since any function $f\in F_{d,r}$ has an absolutely convergent Fourier series, by letting $Q_{d,n}$ being the QMC rule using $P_{d,m}^1$ (or $P_{d,m}^2$) for some $m\geq 2$, 
it follows from Corollary~\ref{cor:exponential_sum} that, with $n$ equal to $\sum_{p\in \PP_m}p^2$,
\begin{align}\label{ineq_det_bound}
\left| I_d(f)-Q_{d,n}(f)\right| & = \left| I_d(f)-\frac{1}{|P_{d,m}^1|}\sum_{p\in \PP_m}\sum_{h=0}^{p^2-1}f(\bsx_h^{(p)})\right| \nonumber \\
        & = \left| \hat{f}(\bszero)-\frac{1}{|P_{d,m}^1|}\sum_{p\in \PP_m}\sum_{h=0}^{p^2-1}\sum_{\bsk\in \ZZ^d}\hat{f}(\bsk)\exp\left( 2\pi \ri \bsk\cdot \bsx_h^{(p)}\right)\right| \nonumber \\
        & = \left| \sum_{\bsk\in \ZZ^d\setminus \{\bszero\}}\hat{f}(\bsk)\frac{1}{|P_{d,m}^1|}\sum_{p\in \PP_m}\sum_{h=0}^{p^2-1}\exp\left( 2\pi \ri \bsk\cdot \bsx_h^{(p)}\right)\right| \nonumber \\
        & \leq \sum_{\bsk\in \ZZ^d\setminus \{\bszero\}}|\hat{f}(\bsk)|\left| \frac{1}{|P_{d,m}^1|}\sum_{p\in \PP_m}\sum_{h=0}^{p^2-1}\exp\left( 2\pi \ri \bsk\cdot \bsx_h^{(p)}\right)\right| \nonumber \\
        & \leq \frac{1}{m}\sum_{\bsk\in \ZZ^d\setminus \{\bszero\}}|\hat{f}(\bsk)|\left( 4\width(\supp(\bsk))+\frac{8}{c_{\PP}}\min_{j\in \supp(\bsk)}\log |k_j|\right) \\
        & \leq \frac{12}{c_{\PP}m}\sum_{\bsk\in \ZZ^d\setminus \{\bszero\}}|\hat{f}(\bsk)|\max\left(\width(\supp(\bsk)),\min_{j\in \supp(\bsk)}\log |k_j|\right) \nonumber \\
        & \leq \frac{12}{c_{\PP}m}\|f\|. \nonumber
    \end{align}
    This leads to an upper bound on the worst-case error of
    \[ e^{\wor}(F_{d,r},Q_{d,n})\leq \frac{12}{c_{\PP}m}.\]
    Therefore, in order to make $e^{\wor}(F_{d,r},Q_{d,n})$ less than or equal to $\varepsilon \in (0,1)$, it suffices to choose $m=\lceil 12c_{\PP}^{-1}\varepsilon^{-1}\rceil$ and we have
    \[  n(\varepsilon, d, F_{d,r})\leq \sum_{p\in \PP_{\lceil 12c_{\PP}^{-1}\varepsilon^{-1}\rceil}}p^2 \leq C_{\PP}\frac{\lceil 12c_{\PP}^{-1}\varepsilon^{-1}\rceil}{\log \lceil 12c_{\PP}^{-1}\varepsilon^{-1}\rceil}\times \left(\lceil 12c_{\PP}^{-1}\varepsilon^{-1}\rceil\right)^2, \]
    from which the result follows.
\end{proof}

\subsection{Numerical integration in the space $F_{d,r_3}$}

This subsection is devoted to proving the following result where the Fourier weights $r$ are of product form.
\begin{theorem}\label{thm:main4}
Consider the space \eqref{eq:subspace_wiener} with $r(\bsk) \geq r_3(\bsk) = \prod_{j=1}^d \max(1, \log |k_j|)$. There exists a constant $C>0$ such that, for any $d\in \NN$ and $\varepsilon \in (0,1)$, we have
    \[ n^{\wor}(\varepsilon, d, F_{d, r})\leq C d^3 \varepsilon^{-3}/(\log (d\varepsilon^{-1} )). \]
\end{theorem}

\begin{proof}
Using the same setup as in the proof of Theorem~\ref{thm:main2} we have for any $f \in F_{d,r}$ from \eqref{ineq_det_bound} that
\begin{align*}
|I_d(f) - Q_{n,d}(f)| \le & \frac{1}{m} \sum_{\bsk \in \ZZ^d \setminus \{\bszero\}} |\hat{f}(\bsk)| r_3(\bsk)  \max_{\bsell \in \ZZ^d \setminus \{\bszero\}} \frac{ 4\width(\supp(\bsell))+\frac{8}{c_{\PP}}\min_{j\in \supp(\bsell)}\log |\ell_j|}{r_3(\bsell)} \\ \le & \frac{12}{c_{\PP} m} \|f\|  \max_{\bsell \in \ZZ^d \setminus \{\bszero\}} \frac{\max( \width(\supp(\bsell)), \min_{j\in \supp(\bsell)}\log |\ell_j|)}{r_3(\bsell)}.
\end{align*}

We have
\begin{equation*}
\max_{\bsell \in \ZZ^d \setminus \{\bszero\}} \frac{\width(\supp(\bsell))}{r_3(\bsell)} \le d
\end{equation*}
and
\begin{equation*}
\max_{\bsell \in \ZZ^d \setminus \{\bszero\}} \frac{\min_{j\in \supp(\bsell)}\log |\ell_j|}{r_3(\bsell)} \le 1.
\end{equation*}
This leads to an upper bound on the worst-case error as
\[ e^{\wor}(F_{d,r},Q_{d,n})\leq \frac{12 d}{c_{\PP}m}.\]
    Therefore, in order to make $e^{\wor}(F_{d,r_3},Q_{d,n})$ less than or equal to $\varepsilon \in (0,1)$, it suffices to choose $m=\lceil 12d c_{\PP}^{-1}\varepsilon^{-1}\rceil$ and we have
    \[  n(\varepsilon, d, F_{d,r})\leq \sum_{p\in \PP_{\lceil 12d c_{\PP}^{-1}\varepsilon^{-1}\rceil}}p^2 \leq C_{\PP}\frac{\lceil 12 d c_{\PP}^{-1}\varepsilon^{-1}\rceil}{\log \lceil 12 d c_{\PP}^{-1}\varepsilon^{-1}\rceil}\times \left(\lceil 12 d c_{\PP}^{-1}\varepsilon^{-1}\rceil\right)^2, \]
    from which the result follows.
\end{proof}


\subsection{Existence result}

The previous two results were based on explicit constructions of point sets. As comparison, in this subsection, we employ a probabilistic argument using Hoeffding's theorem to obtain the existence of point sets such that $n^{\wor}(\varepsilon, d, F_{d, r_4})$ is bounded above.
This proof technique has been used already in similar contexts, like, e.g., in \cite{HNWW01} to prove polynomial tractability of the star discrepancy.
Here we choose the Fourier weights $r_4$ of product form, and we use the slowest possible decay for which the argument still yields polynomial tractability.

Hoeffding's inequality asserts the following: Let $z_0, z_1, \ldots, z_{n-1}$ be independent random variables such that $a_i \le z_i \le b_i$ almost surely. Consider the sum of these random variables, $S_n = z_0 + z_1 + \cdots + z_{n-1}$. Then Hoeffding's theorem states that, for all $t > 0$ we have
\begin{equation*}
\mathbb{P}\left(|S_n - \mathbb{E}(S_n)| \ge t \right) \ge 2\, \mathrm{e}^{-\frac{2 t^2}{\sum_{h=0}^{n-1} (b_h - a_h)^2}},
\end{equation*}
where $\mathbb{E}(S_n)$ is the expected value of $S_n$.

\begin{theorem}\label{thm:main5}
Consider the space \eqref{eq:subspace_wiener} with $r(\bsk) \geq r_4(\boldsymbol{k}) = \prod_{j=1}^d \max(1, |k_j|)$. There exists a constant $C > 0$ such that for any $d \in \mathbb{N}$ and $\varepsilon \in (0,1)$ we have
\begin{equation*}
n^{\wor}(\varepsilon, d, F_{d, r}) \le C \frac{d}{\varepsilon^2}  \log \left(1 + \frac{1}{\varepsilon} \right).
\end{equation*}
\end{theorem}

\begin{proof}
\sloppy Monte Carlo integration approximates the integral $I_d(f)$ using the point set $P_{d,n} = \{\bsx_0, \bsx_1, \ldots, \bsx_{n-1}\} \subset [0,1)^d$ by
\[ Q(f; P_{d,n})=\frac{1}{n}\sum_{h=0}^{n-1}f(\bsx_h),\]
where $\bsx_0,\ldots,\bsx_{n-1}$ are independently and randomly chosen from the uniform distribution over $[0,1)^d$.
For any function $f\in F_{d,r}$, we have 
\[ \EE[Q(f;P_{d,n})] = \frac{1}{n}\sum_{h=0}^{n-1}\EE[f(\bsx_h)] = I_d(f), \]
and also, for any $\bsx\in [0,1)^d$, we have
\[
    |f(\bsx)| = \left| \sum_{\bsk\in \ZZ^d}\hat{f}(\bsk)e^{2\pi \ri \bsk\cdot \bsx}\right|
    \leq \sum_{\bsk\in \ZZ^d}\left| \hat{f}(\bsk) \right|
    \leq \sum_{\bsk\in \ZZ^d}\left| \hat{f}(\bsk)\right| r_4(\bsk)
    \leq \|f\|.
\]

For $\bsk \in \mathbb{Z}^d$ let $e_{\bsk}(\bsx) = \exp(2\pi \ri \bsk \cdot \bsx)$. For a given $\delta > 0$ and $\bsk \in \mathbb{Z}^d \setminus \{\bszero\} $, let
\begin{align*}
\Delta_{\bsk}(\delta) = & \left\{P_{d,n} \subset [0,1)^d\: \mid \: | Q(e_{\bsk}; P_{d,n}) - I_d(e_{\bsk}) | \ge \delta r_4(\bsk) \right\} \\ = &  \left\{P_{d,n} \subset [0,1)^d\: \mid \: | Q(e_{\bsk}; P_{d,n}) | \ge \delta r_4(\bsk) \right\}.
\end{align*}
Now we can use Hoeffding's inequality with $z_h = f(\bsx_h)$, $a_h = - \|f\|$, $b_h = \|f\|$ and $t = n \varepsilon$, for any $\varepsilon>0$, we have
\[ \PP\left[\{P_{d,n} \subset [0,1)^d\: \mid \:  |Q(f; P_{d,n})-I_d(f)|\geq \varepsilon \} \right]  \le  2\exp\left(-\frac{n\varepsilon^2}{2\|f\|^2}\right). \]
Since $\|e_{\bsk}\| = r_4(\bsk)$, Hoeffding's inequality implies that
\begin{equation*}
    \PP[\Delta_{\bsk}(\delta) ] \le 2\exp \left(- \frac{n \delta^2}{2} \right).
\end{equation*}

Let
\begin{equation*}
\Acal = \left\{\bsk \in \mathbb{Z}^d \setminus \{\bszero\}\: \mid \: r_4(\bsk) \le \frac{1}{\delta} \right\}.
\end{equation*}
Then the union bound implies that
\begin{equation*}
\PP\left[ \bigcup_{\bsk \in \Acal} \Delta_{\bsk}(\delta) \right] \le 2|\Acal| \exp\left( - \frac{n \delta^2}{2} \right).
\end{equation*}
For $\bsk\neq \bszero$, let
\begin{align*}
\overline{\Delta}_{\bsk}(\delta) = & \left\{P_{d,n} \subset [0,1)^d\: \mid \: | Q(e_{\bsk}; P_{d,n}) - I_d(e_{\bsk}) | < \delta r_4(\bsk)  \right\} \\ = &  \left\{P_{d,n} \subset [0,1)^d\: \mid \: | Q(e_{\bsk}; P_{d,n}) | < \delta r_4(\bsk)  \right\}.
\end{align*}
Then
\begin{equation*}
\PP \left[ \bigcap_{\bsk \in \Acal} \overline{\Delta}_{\bsk}(\delta)  \right] > 1 - 2|\Acal| \exp\left(-\frac{n \delta^2}{2} \right).
\end{equation*}
Thus, if $1 - 2|\Acal| \exp\left(- \frac{n \delta^2}{2}\right) \ge 0$, then there exists a point set $P_{d,n}^\ast \subset [0,1)^d$ of size $n$ such that
\begin{equation*}
|Q(e_{\bsk}; P_{d,n}^\ast) - I_d(e_{\bsk}) | < \delta r_4(\bsk), \quad \forall \bsk \in \Acal,
\end{equation*}
i.e., this point set $P_{d,n}^\ast$ numerically integrates all $e_{\bsk}$, $\bsk \in \Acal$, simultaneously with the specified error. Now $1 - 2|\Acal| \exp\left( - \frac{n \delta^2}{2} \right) \ge 0$ is equivalent to
\[
|\Acal| \le \frac{1}{2}\exp\left( \frac{n \delta^2}{2} \right) ,
\]
and this in turn is equivalent to
\[
 \frac{n\delta^2}{2}\geq \log(2|\Acal|).
\]

Let now $f \in F_{d,r}$ with $r\geq r_4$. Then
\begin{align*}
\left| Q(f; P_{d,n}^\ast) - I_d(f) \right| = &  \left| \sum_{\bsk \in \mathbb{Z}^d \setminus \{\bszero\} } \widehat{f}(\bsk) \left( Q(e_{\bsk}; P_{d,n}^\ast) - I_d(e_{\bsk}) \right) \right| \\ \le  & \sum_{\bsk \in \mathbb{Z}^d \setminus \{\bszero\}  } \left|\widehat{f}(\bsk) \right| \left| Q(e_{\bsk}; P_{d,n}^\ast) - I_d(e_{\bsk}) \right| \\ \le & \delta \sum_{\bsk \in \Acal} | \widehat{f}(\bsk) | r_4(\bsk) +  \sum_{\bsk \in ( \mathbb{Z}^d \setminus |\Acal| ) \setminus \{\bszero\} } |\widehat{f}(\bsk)| \\ \le &  \delta \sum_{\bsk \in \Acal} | \widehat{f}(\bsk) | r_4(\bsk) +  \max_{\bsk \in (\mathbb{Z}^d \setminus \Acal ) \setminus \{\bszero\} } \frac{1}{r_4(\bsk)} \sum_{\bsk \in (\mathbb{Z}^d \setminus |\Acal| ) \setminus \{\bszero\} } |\widehat{f}(\bsk)| r_4(\bsk) \\ \le & \delta \|f\|.
\end{align*}
Thus
\begin{equation}\label{eq:worst-case}
\sup_{f \in F_{d,r}: \|f\| \le 1} \left| Q(f; P_{d,n}^\ast) - I_d(f) \right| \le \delta.
\end{equation}

We now estimate the size of the set $\Acal$, given by
\begin{align*}
\Acal = \Acal(\delta,r_4,d) = & \left\{ \bsk \in \mathbb{Z}^d \setminus \{\bszero\}\: \mid \: r_4(\bsk) \le \frac{1}{\delta} \right\}.
\end{align*}
Although volume estimates are available in the literature (see, e.g., \cite{CD16,KSU15}), for completeness, we prove, by induction on $d$, that
\[|\Acal(\delta,r_4,d)|\leq \frac{4^d}{\delta}\left( 1+\log \frac{1}{\delta}\right)^{d-1} \] 
holds for any $\delta\in (0,1].$

For the case $d=1$, we have
\[ \left|\Acal(\delta, r_4,1)\right| \leq 1+2\left\lfloor \frac{1}{\delta}\right\rfloor \leq 1+\frac{2}{\delta}\leq \frac{4}{\delta}.\]
For $d>1$, assume that
\[|\Acal(\delta,r_4,d-1)|\leq \frac{4^{d-1}}{\delta}\left(1+\log\frac{1}{\delta}\right)^{d-2} \] 
for any $\delta\in (0,1]$. Then we have
\begin{align*}
    \left|\Acal(\delta, r_4,d)\right| & = \left|\Acal(\delta, r_4,d-1)\right|+\left|\left\{ \bsk \in \mathbb{Z}^d \setminus \{\bszero\}\: \mid \: k_d\neq 0\quad \text{and}\quad r_4(\bsk) \le \frac{1}{\delta} \right\}\right| \\
    & = \left|\Acal(\delta, r_4,d-1)\right|+2\sum_{k_d=1}^{\lfloor 1/\delta\rfloor}\left|\left\{ \bsk \in \mathbb{Z}^{d-1} \: \mid \: r_4(\bsk) \le \frac{1}{k_d\delta} \right\}\right| \\
    & = \left|\Acal(\delta, r_4,d-1)\right|+2\sum_{k_d=1}^{\lfloor 1/\delta\rfloor}\left(1+\left|\Acal(k_d\delta, r_4,d-1)\right|\right)\\
    & \leq \frac{4^{d-1}}{\delta}\left(1+\log\frac{1}{\delta}\right)^{d-2}+2\left\lfloor \frac{1}{\delta}\right\rfloor+2\sum_{k_d=1}^{\lfloor 1/\delta\rfloor}\frac{4^{d-1}}{k_d\delta}\left(1+\log\frac{1}{k_d\delta}\right)^{d-2} \\
    & \leq \frac{4^{d-1}}{\delta}\left(1+\log\frac{1}{\delta}\right)^{d-2}+\frac{2}{\delta}+\frac{2\cdot 4^{d-1}}{\delta}\left(1+\log\frac{1}{\delta}\right)^{d-2}\sum_{k_d=1}^{\lfloor 1/\delta\rfloor}\frac{1}{k_d}\\
    & \leq \frac{4^{d-1}}{\delta}\left(1+\log\frac{1}{\delta}\right)^{d-2}+\frac{2}{\delta}+\frac{2\cdot 4^{d-1}}{\delta}\left(1+\log\frac{1}{\delta}\right)^{d-1}\\
    & \leq \frac{4^{d-1}+2+2\cdot 4^{d-1}}{\delta}\left(1+\log\frac{1}{\delta}\right)^{d-1}\\ 
    & \leq \frac{4^d}{\delta}\left(1+\log\frac{1}{\delta}\right)^{d-1},
\end{align*}
where, in the first inequality, we used the induction assumption, and in the third inequality, we used
\[ \sum_{k_d=1}^{\lfloor 1/\delta\rfloor}\frac{1}{k_d} \leq 1+\int_1^{\lfloor 1/\delta\rfloor}\frac{1}{x}\rd x=1+\log \left\lfloor \frac{1}{\delta}\right\rfloor\leq 1+\log\frac{1}{\delta}.\]

Now if we choose $n$ such that
\[ \frac{ n\delta^2}{2} \ge \log \left(2 \cdot \frac{4^d}{\delta}\left( 1+\log \frac{1}{\delta}\right)^{d-1} \right) =\log 8 +\log\frac{1}{\delta} + (d-1) \log \left(4+4\log\frac{1}{\delta}\right) \]
holds, then there exists a point set $P_{d,n}^*$ which satisfies \eqref{eq:worst-case}. This means that the information complexity is bounded by
\[ n^{\wor}(\varepsilon, d, F_{d,r})\leq C d\varepsilon^{-2}\log(1+\varepsilon^{-1}), \]
for some constant $C>0$.
\end{proof}

\begin{remark}
    In \cite{JUV23,KPUU23}, multivariate $L_p$ approximation for the space \eqref{eq:subspace_wiener} with
    \[ r(\bsk)=\prod_{j=1}^{d}(1+|k_j|)^\alpha, \quad \text{or}\quad r(\bsk)=\prod_{j=1}^{d}\max(1,|k_j|^\alpha), \]
    for $\alpha>0$, has been studied as a special case of more general frameworks. Since the integration problem is no harder than the approximation problem,
    it follows from \cite[Corollary~4.4]{JUV23} that a convergence rate for the worst-case error of order $n^{-\alpha-1/2}(\log n)^{c(d)}$ with an exponent $c(d)\geq (d-1)\alpha+1/2$ is achieved.
    For $\alpha=1$, this order (in terms of $n$) is superior to what is possible by Monte Carlo integration. However, since our focus is on tractability, whereas theirs is not, further investigation is needed for a direct comparison. This is similar to the star discrepancy problem, where the best-known rate is $(\log n)^{d-1}/n$, yet the bound $\sqrt{d/n}$ is more favorable in terms of the dependence on the dimension $d$.
\end{remark}

\begin{remark}
    Consider the space \eqref{eq:subspace_wiener} with $r(\bsk) = \max(1, |k_1|,\ldots,|k_d|)$, which is larger than the space \eqref{eq:subspace_wiener} with $r(\bsk) = r_4(\bsk)$ if $d\geq 2$. Since 
    \[ \left| \left\{ \bsk \in \mathbb{Z}^d \setminus \{\bszero\}\: \mid \: \max(1, |k_1|,\ldots,|k_d|) \le \frac{1}{\delta} \right\} \right| = \left( 1+2\left\lfloor \frac{1}{\delta}\right\rfloor\right)^d, \]
    applying the same argument as in the proof of Theorem~\ref{thm:main5} leads to the following bound on the information complexity:
\[ n^{\wor}(\varepsilon, d, F_{d,r})\leq C d\varepsilon^{-2}\log(1+\varepsilon^{-1}), \]
for some constant $C>0$.
\end{remark}


\section{Strong tractability in the randomized setting}\label{sec:str_tractable_ran}

Finally, we show the following upper bound in the randomized setting.
\begin{theorem}\label{thm:main3}
Consider the space \eqref{eq:subspace_wiener} with $r(\bsk) \geq r_1(\boldsymbol{k}) = \max(1, \log \min_{j \in \supp(\bsk)} |k_j|)$. There exists a constant $C>0$ such that, for any $d\in \NN$ and $\varepsilon \in (0,1)$, we have
    \[ n^{\ran}(\varepsilon, d, F_{d,r})\leq C \varepsilon^{-1}. \]
\end{theorem}

Note that the $\varepsilon$-exponent we obtain is 1. As mentioned earlier, this is better than the $\varepsilon$-exponent of 2 expected from the standard Monte Carlo method.

Our proof is again constructive in the sense that we provide an explicit randomized QMC rule that attains the desired randomized error bound. The randomized QMC rule considered here is similar to those studied in \cite{DGS22,KKNU19}. For an integer $m\geq 2$, we randomly pick $p\in \PP_m$ with uniform distribution and then randomly pick $\bsz=(z_1,\ldots,z_d)\in \{1,\ldots,p-1\}^d$ with uniform distribution. Once $p$ and $\bsz$ are fixed, our algorithm becomes deterministic and is given by the QMC rule using the point set
\[ P_{p,\bsz}=\left\{ \bsx_{h}^{(p,\bsz)}=\left( \left\{ \frac{hz_1}{p}\right\},\ldots,\left\{ \frac{hz_d}{p}\right\}\right)\; \mid\; 0\leq h<p\right\},\]
that is called a rank-1 lattice point set. The following character property of rank-1 lattice point sets is well-known: for any $\bsk\in \ZZ^d$, we have
\[ \frac{1}{p}\sum_{h=0}^{p-1}\exp\left(2\pi \ri \bsk\cdot \bsx_h^{(p,\bsz)}\right)=\begin{cases} 1 & \text{if $\bsk\cdot \bsz\equiv 0 \pmod p$,} \\ 0 & \text{otherwise.}\end{cases}\]
We refer to \cite[Lemmas~4.2 \& 4.3]{DHP15} for the proof.

Now we are ready to prove Theorem~\ref{thm:main3}.
\begin{proof}[Proof of Theorem~\ref{thm:main3}]
    Since any function $f\in F_{d,r}$ has an absolutely convergent Fourier series, an argument similar to the proof of Theorem~\ref{thm:main2} leads to
    \begin{align*}
        \EE_{\omega}\left[ \left| I_d(f)-Q_{d,\omega}(f)\right| \right] & = \frac{1}{|\PP_m|}\sum_{p\in \PP_m}\frac{1}{(p-1)^d}\sum_{\bsz\in \{1,\ldots,p-1\}^d} \left| I_d(f)-\frac{1}{p}\sum_{h=0}^{p-1}f(\bsx_h^{(p,\bsz)})\right|\\
        & = \frac{1}{|\PP_m|}\sum_{p\in \PP_m}\frac{1}{(p-1)^d}\sum_{\bsz\in \{1,\ldots,p-1\}^d} \left| \sum_{\bsk\in \ZZ^d\setminus \{\bszero\}}\hat{f}(\bsk)\frac{1}{p}\sum_{h=0}^{p-1}\exp\left(2\pi \ri \bsk\cdot \bsx_h^{(p,\bsz)}\right)\right|\\
        & \leq \frac{1}{|\PP_m|}\sum_{p\in \PP_m}\frac{1}{(p-1)^d}\sum_{\bsz\in \{1,\ldots,p-1\}^d} \sum_{\bsk\in \ZZ^d\setminus \{\bszero\}}|\hat{f}(\bsk)| \left|\frac{1}{p}\sum_{h=0}^{p-1}\exp\left(2\pi \ri \bsk\cdot \bsx_h^{(p,\bsz)}\right)\right|\\
        & = \sum_{\bsk\in \ZZ^d\setminus \{\bszero\}}|\hat{f}(\bsk)| \frac{1}{|\PP_m|}\sum_{p\in \PP_m}\frac{1}{(p-1)^d}\sum_{\substack{\bsz\in \{1,\ldots,p-1\}^d\\ \bsk\cdot \bsz\equiv 0 \pmod p}} 1,
    \end{align*}
    where the last equality follows from the character property of rank-1 lattice point sets.

    Let us consider the condition $\bsk\cdot \bsz\equiv 0 \pmod p$. If $p\mid \bsk$, i.e., if all the components of $\bsk$ are multiples of $p$, then this condition holds for any $\bsz\in \{1,\ldots,p-1\}^d$. On the other hand, if there exists at least one component $k_j$ that is not a multiple of $p$, as the condition is equivalent to
    \[ k_jz_j\equiv -\sum_{\substack{i=1\\ i\neq j}}^{d}k_iz_i \pmod p, \]
    there exists at most one $z_j\in \{1,\ldots,p-1\}$ for any vector $(z_i)_{i\neq j}$. Thus, it follows that
    \begin{align*}
        \frac{1}{|\PP_m|}\sum_{p\in \PP_m}\frac{1}{(p-1)^d}\sum_{\substack{\bsz\in \{1,\ldots,p-1\}^d\\ \bsk\cdot \bsz\equiv 0 \pmod p}} 1 & \leq \frac{1}{|\PP_m|}\sum_{\substack{p\in \PP_m\\ p\mid \bsk}}1+\frac{1}{|\PP_m|}\sum_{\substack{p\in \PP_m\\ p\nmid \bsk}}\frac{1}{p-1}\\
        & \leq \frac{2}{|\PP_m|\log m}\min_{j\in \supp(\bsk)}\log |k_j|+\frac{2}{m}\\
        & \leq \frac{2}{c_{\PP}m}\min_{j\in \supp(\bsk)}\log |k_j|+\frac{2}{m}\\
        & \leq \frac{4}{c_{\PP}m}\max\left(1,\min_{j\in \supp(\bsk)}\log |k_j|\right)=\frac{4}{c_{\PP}m} r_1(\bsk),
    \end{align*}
    where we have used the result on the sum $\sum_{\substack{p\in \PP_m\\ p\mid \bsk}}1$ shown in the proof of Corollary~\ref{cor:exponential_sum}. 
    
    Using this estimate, we obtain
    \[ \EE_{\omega}\left[ \left| I_d(f)-Q_{d,\omega}(f)\right| \right] \leq \frac{4}{c_{\PP}m}\sum_{\bsk\in \ZZ^d\setminus \{\bszero\}}|\hat{f}(\bsk)| r_1(\bsk) \leq \frac{4}{c_{\PP}m}\|f\|. \]
    Since the cardinality of our randomized algorithm is obviously bounded above by $m$, the information complexity in the randomized setting is bounded above by
    \[ n^{\ran}(\varepsilon, d, F_{d,r}) \leq \left\lceil \frac{4}{c_{\PP}}\varepsilon^{-1}\right\rceil,\]
    for any $\varepsilon\in (0,1)$ and $d\in \NN$, which completes the proof.
\end{proof}

\bibliographystyle{plain}
\bibliography{ref.bib}

\begin{thebibliography}{10}

\bibitem{CH24}
L.~Chen and H.~Jiang.
\newblock On the information complexity for integration in subspaces of the
  {W}iener algebra.
\newblock {\em J. Complexity}, 81:Paper No. 101819, 9, 2024.

\bibitem{CD16}
A.~Chernov and D.~D\~ung.
\newblock New explicit-in-dimension estimates for the cardinality of
  high-dimensional hyperbolic crosses and approximation of functions having
  mixed smoothness.
\newblock {\em J. Complexity}, 32(1):92--121, 2016.

\bibitem{D14}
J.~Dick.
\newblock Numerical integration of {H}\"older continuous, absolutely convergent
  {F}ourier, {F}ourier cosine, and {W}alsh series.
\newblock {\em J. Approx. Theory}, 183:14--30, 2014.

\bibitem{DGS22}
J.~Dick, T.~Goda, and K.~Suzuki.
\newblock Component-by-component construction of randomized rank-1 lattice
  rules achieving almost the optimal randomized error rate.
\newblock {\em Math. Comp.}, 91(338):2771--2801, 2022.

\bibitem{DGPW17}
J.~Dick, D.~Gomez-Perez, F.~Pillichshammer, and A.~Winterhof.
\newblock Digital inversive vectors can achieve polynomial tractability for the
  weighted star discrepancy and for multivariate integration.
\newblock {\em Proc. Amer. Math. Soc.}, 145(8):3297--3310, 2017.

\bibitem{DHP15}
J.~Dick, A.~Hinrichs, and F.~Pillichshammer.
\newblock Proof techniques in quasi--{M}onte {C}arlo theory.
\newblock {\em J. Complexity}, 31(3):327--371, 2015.

\bibitem{DKP22}
J.~Dick, P.~Kritzer, and F.~Pillichshammer.
\newblock {\em Lattice rules---numerical integration, approximation, and
  discrepancy}, volume~58 of {\em Springer Series in Computational
  Mathematics}.
\newblock Springer, Cham, 2022.

\bibitem{DP15}
J.~Dick and F.~Pillichshammer.
\newblock The weighted star discrepancy of {K}orobov's {$p$}-sets.
\newblock {\em Proc. Amer. Math. Soc.}, 143(12):5043--5057, 2015.

\bibitem{G23}
T.~Goda.
\newblock Polynomial tractability for integration in an unweighted function
  space with absolutely convergent {F}ourier series.
\newblock {\em Proc. Amer. Math. Soc.}, 151(9):3925--3933, 2023.

\bibitem{HNWW01}
S.~Heinrich, E.~Novak, G.~W. Wasilkowski, and H.~Wo\'zniakowski.
\newblock The inverse of the star-discrepancy depends linearly on the
  dimension.
\newblock {\em Acta Arith.}, 96(3):279--302, 2001.

\bibitem{HW81}
L.~K. Hua and Y.~Wang.
\newblock {\em Applications of number theory to numerical analysis}.
\newblock Springer-Verlag, Berlin-New York, 1981.

\bibitem{JUV23}
T.~Jahn, T.~Ullrich, and F.~Voigtlaender.
\newblock Sampling numbers of smoothness classes via {$\ell^1$}-minimization.
\newblock {\em J. Complexity}, 79:Paper No. 101786. 35, 2023.

\bibitem{K63}
N.~M. Korobov.
\newblock {\em {Number-Theoretic Methods in Approximate Analysis}}.
\newblock Fizmatgiz, Moscow, 1963.

\bibitem{K23}
D.~Krieg.
\newblock Tractability of sampling recovery on unweighted function classes.
\newblock {\em Proc. Amer. Math. Soc. Ser. B}, 11:115--125, 2024.

\bibitem{KPUU23}
D.~Krieg, K.~Pozharska, M.~Ullrich, and T.~Ullrich.
\newblock Sampling recovery in $l_2$ and other norms.
\newblock {\em arXiv preprint arXiv:2305.07539}, 2023.

\bibitem{KV23}
D.~Krieg and J.~Vyb\'iral.
\newblock New lower bounds for the integration of periodic functions.
\newblock {\em J. Fourier Anal. Appl.}, 29(4):Paper No. 41, 26, 2023.

\bibitem{KKNU19}
P.~Kritzer, F.~Y. Kuo, D.~Nuyens, and M.~Ullrich.
\newblock Lattice rules with random {$n$} achieve nearly the optimal {$\mathcal
  O(n^{-\alpha-1/2})$} error independently of the dimension.
\newblock {\em J. Approx. Theory}, 240:96--113, 2019.

\bibitem{KSU15}
T.~K\"uhn, W.~Sickel, and T.~Ullrich.
\newblock Approximation of mixed order {S}obolev functions on the {$d$}-torus:
  asymptotics, preasymptotics, and {$d$}-dependence.
\newblock {\em Constr. Approx.}, 42(3):353--398, 2015.

\bibitem{NW08}
E.~Novak and H.~Wo\'zniakowski.
\newblock {\em Tractability of multivariate problems. {V}ol. 1: {L}inear
  information}, volume~6 of {\em EMS Tracts in Mathematics}.
\newblock European Mathematical Society (EMS), Z\"urich, 2008.

\bibitem{NW10}
E.~Novak and H.~Wo\'zniakowski.
\newblock {\em Tractability of multivariate problems. {V}olume {II}: {S}tandard
  information for functionals}, volume~12 of {\em EMS Tracts in Mathematics}.
\newblock European Mathematical Society (EMS), Z\"urich, 2010.

\bibitem{RS62}
J.~B. Rosser and L.~Schoenfeld.
\newblock Approximate formulas for some functions of prime numbers.
\newblock {\em Illinois J. Math.}, 6:64--94, 1962.

\bibitem{SW98}
I.~H. Sloan and H.~Wo\'zniakowski.
\newblock When are quasi-{M}onte {C}arlo algorithms efficient for
  high-dimensional integrals?
\newblock {\em J. Complexity}, 14(1):1--33, 1998.

\bibitem{TWW88}
J.~F. Traub, G.~W. Wasilkowski, and H.~Wo\'zniakowski.
\newblock {\em Information-based complexity}.
\newblock Academic Press, Inc., Boston, MA, 1988.

\end{thebibliography}

\end{document}